\documentclass[onecolumn, draftcls,12pt]{IEEEtran}

\IEEEoverridecommandlockouts
\usepackage[latin1]{inputenc}
\usepackage{graphicx}
\usepackage{amssymb,amsmath,amsfonts,amsthm}
\usepackage{algorithm,algorithmic}
\usepackage{cite}
\usepackage{float}
\usepackage{graphics}
\usepackage{subfigure}
\usepackage{xspace}
\usepackage[usenames,dvipsnames]{color}
\usepackage{epsfig, hyperref}
\usepackage{verbatim}
\usepackage{url,changebar,bm,xspace,dsfont}
\usepackage[T1]{fontenc}
\usepackage{setspace} 
\usepackage{stfloats}

\def\maxlink{\mathop{\mathrm{max-link}}}
\def\minpow{\mathop{\mathrm{min-power}}}

\floatstyle{ruled}
\newfloat{model}{H}{mod}
\floatname{model}{\footnotesize Model}
\newfloat{notatio}{H}{not}
\floatname{notatio}{\footnotesize Notation}

\newtheorem{lemma}{Lemma}
\newtheorem{proposition}{Proposition}

\newtheorem{coro}{Corollary}
\newtheorem{remark}{Remark}

\begin{document}

\title{A Buffer-aided Successive Opportunistic Relay Selection Scheme with Power Adaptation and Inter-Relay Interference Cancellation for Cooperative Diversity Systems}

\author{Nikolaos Nomikos,~\IEEEmembership{Student Member,~IEEE,} Themistoklis Charalambous,~\IEEEmembership{Member,~IEEE,}  Ioannis Krikidis,~\IEEEmembership{Senior Member,~IEEE,} \\  Dimitrios N. Skoutas,~\IEEEmembership{Member,~IEEE,} Demosthenes Vouyioukas,~\IEEEmembership{Member,~IEEE,} \\ and Mikael Johansson,~\IEEEmembership{Member,~IEEE.} 
\thanks{N. Nomikos, D. Vouyioukas and D. N. Skoutas are with the Department of Information and Communication Systems Engineering, University of the Aegean, Karlovassi 83200, Samos, Greece (E-mails: {\tt \{nnomikos,dvouyiou,d.skoutas\}@aegean.gr}).}
\thanks{T. Charalambous and M. Johansson are with the Automatic Control Lab, Electrical Engineering Department and ACCESS Linnaeus Center, Royal Institute of Technology (KTH), Stockholm, Sweden (E-mails:  {\tt \{themisc,mikaelj\}@kth.se}). }
\thanks{I. Krikidis is with the ECE Department, University of Cyprus, Nicosia 1678 (E-mail:  {\tt krikidis@ucy.ac.cy}).}
\thanks{Preliminary results of this work have been published in \cite{2013:PIMRC}. This paper extends the work further through detailed algorithm description, discussion on channel state information schemes, outage probability and diversity analysis, extra illustrative examples and performance evaluation scenarios based on various system parameters.}
}

\maketitle

%
%
\vspace{-0.0cm}
\begin{abstract}
In this paper, we present a relaying scheme which combines the spectral efficiency of successive opportunistic relaying with the robustness of single-link selection. More specifically, we propose the $\minpow$ scheme that minimizes the total energy expenditure per time slot under an inter-relay interference (IRI) cancellation scheme. It is the first time that interference cancellation is combined with buffer-aided relays and power adaptation to mitigate the IRI and minimize the energy expenditure. The new relay selection policy is analyzed in terms of outage probability and diversity by modeling the evolution of the relay buffers as a Markov Chain. We construct the state transition matrix of the Markov Chain and obtain its stationary distribution, which in turn, yields the outage probability. The proposed scheme outperforms relevant state-of-the-art relay selection schemes in terms of throughput, diversity and energy efficiency, as demonstrated via representative numerical examples. 
\end{abstract}


\IEEEpeerreviewmaketitle

\vspace{-0.0cm}
%
%
\section{Introduction}

Cooperative relaying is an efficient technique to combat fading and path-loss effects in wireless systems. It enables multiple nodes to create virtual multiple-input multiple-output (MIMO) configurations in order to provide spatial transmit and/or receive diversity to single-antenna destinations \cite{LAN}. Traditional cooperative systems are characterized by the half-duplex constraint that relay nodes cannot receive and transmit data simultaneously, a fact that results in bandwidth loss. In order to overcome this bandwidth limitation, several techniques have been proposed in the literature (see, for example, \cite{DIN}). Among them, the successive relaying scheme in \cite{FAN, chao_ISIT} incorporates multiple relay nodes and proposes a transmission overlap (source-relay, relay-destination) in order to mimic an ideal full-duplex transmission. When IRI is strong (in co-located or clustered relays) it is assumed that it can always be decoded at the affected nodes; this decoded IRI is exploited in a superposition coding scheme that improves the diversity-multiplexing tradeoff (DMT) performance of the system. Regarding the DMT of successive relaying, the authors in \cite{WIC_TCOM} have proved that when perfect decoding at the relays takes place a 3X1 multiple-input single-output (MISO) DMT performance can be achieved. On the other hand, relay selection has been introduced as a spectrally efficient solution that achieves full diversity by requiring only one additional orthogonal channel. In addition, it reduces the complexity of the network since multi-relay schemes rely on distributed space-time codes and the coordination needed among the relays \cite{BLE}. In earlier works, in which relays were assumed to lack data buffers, relay selection was mainly based on the $\max-\min$ criterion and its variations (see, for example, \cite{BLE,KAR,KRI_minmax, BLE1} and references therein). As a result, based on either proactive or reactive criteria, the relay that received the source signal is the same as the one that is subsequently forwarding the signal towards the destination. A variation of those schemes have been proposed in \cite{IKKI_TCOM} where incremental best relay selection was studied. In that case, the direct source-destination link is available and the best relay among the available ones participates in the transmission only when the destination sends a negative acknowledgment.

With the adoption of buffer-aided relays, this coupling is broken, since different relays could be selected for transmission and reception, thus allowing increased degrees of freedom for scenarios where delay insensitive applications take place, or, when delay-aware relay selection is adopted. Buffering at the relay nodes is a promising solution for cooperative networks and motivates the investigation of new protocols and transmission schemes, even though it can result in delayed packet transmission making them inappropriate for some delay sensitive applications \cite{SCH_CM}. The first works that studied the benefits offered by buffer-aided relays, to the best of the authors' knowledge, are \cite{MEHTA, WANG}; the authors in \cite{MEHTA} investigated the performance of ``rateless'' codes in terms of throughput, end-to-end delay and queue stability when buffer-aided relay selection is employed, while \cite{WANG} presented an opportunistic buffered decode-wait-and-forward protocol which exploited mobile relays with buffers to increase the throughput and reduce the delay of the network. Subsequently, Ikhlef \emph{et al.} \cite{IKH2} proposed a novel criterion based on $\max-\max$ Relay Selection (MMRS), in which the relay with the best Source-Relay (SR) link is selected for reception and the relay with the best Relay-Destination (RD) link is selected for transmission. The MMRS implementation with buffers of finite size is additionally discussed. Another work that adopts MMRS is \cite{IKH3}, which aims to recover the half-duplex loss by adopting successive transmissions. As the proposed topology aims to mimic full-duplex relaying, different relays are selected at the same time slot. However, relays are considered isolated and the effect of inter-relay interference is ignored, while relay buffers are never considered to be either full or empty. Also, outage performance is examined for a fixed rate and the achieved capacity for adaptive rate. Krikidis \emph{et al.} \cite{KRICHAR} proposed the $\maxlink$ protocol, which allows all the SR and RD links to enter the competition for the best link through which a signal will be transmitted, thus providing additional freedom in the scheduled transmissions at each time slot. Adaptive link selection, but for merely a single relay setup, is proposed by Zlatanov \emph{et al.} in \cite{ZLAadapt}. The performance potentials of the buffer-aided relaying concept are further investigated in \cite{ZLA2},  where the authors show that half-duplex relaying with storage capabilities outperforms ideal full-duplex relaying.  
 
In the majority of the aforementioned works the main target is outage probability reduction or throughput improvement. The algorithm in \cite{NOMspringer} selects the end-to-end paths based on their achievable capacities and using hierarchical and adaptive modulation the power of the source and the transmitting relay are adjusted  to a target Signal-to-Noise Ratio (SNR). In various works (see, e.g., \cite{FENG, SHENG} and references therein) relays are assumed to be battery operated. Due to this practical assumption, the proposed algorithms select the best relays according to two factors. First, the selected relay should spend the minimum amount of energy to satisfy the outage threshold of the network. Second, by considering the remaining energy of each node, relays that were selected often are protected in order to increase the network lifetime. Optimal energy allocation techniques are presented in \cite{IKKI}, where the effects of co-channel interference at the relays and the destination are studied. A closed-form expression for the outage probability is derived and in addition, under different global and individual energy constraints, resource allocation algorithms are provided in a multi-relay network. Finally, the authors in \cite{CHEN_HUANG}, propose joint relay selection and subcarrier allocation in a multi-user network aiming at total transmission power minimization while still achieving a specific QoS requirement. 

This work proposes a buffer-aided successive opportunistic relaying protocol with Power Adaptation and IRI cancellation. Through the successive nature of this protocol, we aim to recover the half-duplex loss of cooperative relaying. In contrast to other works in the literature, not only we mitigate the detrimental effect of IRI through interference cancellation (IC) with the aid of buffers, if feasible, but we also allow for single link transmissions when IC is not feasible and the target spectral efficiency is not satisfied. More specifically, we propose the $\minpow$ relay selection policy that acts in conjunction with interference cancellation, whenever possible, and adjusts the power levels required accordingly to support the end-to-end communication. The $\minpow$ relay selection policy in terms of outage probability and diversity, is analyzed by modeling the evolution of the relay buffers as a Markov Chain (MC). The construction of the state transition matrix and the related steady state of the MC are studied; then, the derivation of the outage probability is presented. The contribution of this work is twofold.
\begin{itemize}
\item[(i)] Buffer-aided relays and interference cancellation are combined for the first time, thus decoupling the necessity of the receiving relay to transmit in the next time slot, even if the channel is in outage. Hence, an extra degree of freedom is obtained for choosing which relay to transmit, so that IRI cancellation is also achieved, if it is feasible. Moreover, the proposed scheme has a hybrid nature so, when successive transmission is not feasible, a single link transmission is incorporated in our protocol; this is equivalent to the $\maxlink$ protocol \cite{KRICHAR}. Also, the current channel state is the deciding factor in choosing a relay rather than the prediction of the next channel state, as proposed in some works in the literature (see, for example \cite{NOM}).
\item[(ii)] Power adjustment is also included in our scheme. In this way, the total energy expenditure in the network is minimized, as well as the inter-relay interference, thus reducing the outage probability of the network.
\end{itemize}
The $\minpow$ relay selection policy outperforms the current state-of-the-art schemes, with which it is compared. The outage, throughput and energy efficiency performance metrics considered are uniformly improved.

The structure of this paper is as follows. In Section~\ref{sec:model}, we present the system model while Section~\ref{sec:minpow} presents in detail the $\minpow$ relay selection policy proposed herein. Then, a model of this communication scheme and an outage probability analysis is performed in Section~\ref{sec:outage}, while illustrative examples and numerical results are provided in Sections \ref{sec:examples} and~\ref{sec:numerical}, respectively. Finally, conclusions are discussed in Section~\ref{sec:conclusions}.

\vspace{-.2cm}

%
%
\section{System model}\label{sec:model}

We assume a simple cooperative network consisting of one source $S$, one destination $D$ and a cluster $\mathcal{C}$ with $K$ Decode-and-Forward (DF) relays $R_k \in \mathcal{C}$ ($1 \leq k\leq K$). All nodes are characterized by the half-duplex constraint and therefore they cannot transmit and receive simultaneously. A direct link between the source and the destination does not exist and communication can be established only via relays \cite{BLE}. Each relay $R_k$ holds a buffer (data queue) $Q_k$ of length $L_k$ (number of data elements) where it can store source data that has been decoded at the relay and can be forwarded to the destination. The parameter $l_k \in \mathbb{Z}_{+}$, $l_k \in[0, L_k]$ denotes the number of data elements that are stored in buffer $Q_k$; at the beginning, each relay buffer is empty (i.e., $l_k =0$ for all $k$). For simplicity of exposition, we assume that $L_k=L, \forall ~ k\in\{1,2,\ldots, K\}$. 
We denote by $\mathcal{T}$ all the relays for which their buffer is not empty (i.e., $\mathcal{T} = \{R_k: l_k>0 \}$, $\mathcal{T} \subseteq \mathcal{C}$) and hence able to transmit to the destination, and by $\mathcal{A}$ all the relays for which their buffer is not full (i.e., $\mathcal{A} = \{R_k: l_k<L\}$, $\mathcal{A} \subseteq \mathcal{C}$) and they are available to receive a packet from the source. 

Time is considered to be slotted and at each time-slot the source $S$ and (possibly) one of the relays $R_k$ transmit with power $P_S$ and $P_{R_k}$, respectively. The source node is assumed to be saturated (it has always data to transmit) and the information rate is equal to $r_0$.  The retransmission process is based on an Acknowledgment/Negative-Acknowledgment (ACK/NACK) mechanism, in which short-length error-free packets are broadcasted by the receivers (either a relay $R_k$ or the destination $D$) over a separate narrow-band channel in order to inform the network of that packet's reception status. As the relays have buffering capabilities, it is highly probable that the relay selected for transmission will forward a packet received in a previous transmission phase other than the preceding. So, the destination could receive packets in a different order than the one that transmitted by the source and needs to handle them in such a way that information is correctly reconstructured. This can be easily achieved by including a sequence number in each packet, thus allowing the destination to put the packets in order.

All wireless links exhibit fading and we assume an Additive White Gaussian Noise (AWGN) channel. We assume frequency non-selective Rayleigh block fading according to a complex Gaussian distribution with zero mean and variance $\sigma^2_{ij}$ for the $i$ to $j$ link, which means it is constant during one time slot, but changes independently from one slot to another. For simplicity, the variance of the AWGN is assumed to be normalized with zero mean and unit variance. The channel gains are $g_{ij} \triangleq |h_{ij}|^2$ and exponentially distributed \cite[Appendix A]{DTSE}. In addition, we consider a clustered relay topology which offers equivalent average SNR in the SR and RD links ($\sigma^2_{SR_k}=\sigma^2_{R_kD}$). More specifically, the relays are positioned relatively close together based on location-based clustering and through a long-term routing process, variations due to pathloss and shadowing effects are tracked. This model is often assumed in the literature and its statistical analysis was described in \cite{KRI_par}. Also, the works that we compare our scheme with study networks with independent and identical distributed channels (see \cite{BLE, NOM, IKH2, KRICHAR, IKH3}. The power level chosen by the transmitter $i$ is denoted by $P_{i}$. $n_{j}$ denotes the variance of thermal noise at the receiver $j$, which is assumed to be AWGN. 

Since we implement successive relaying, we (may) have concurrent transmissions by the source and one relay taking place at the same time slot. This relaying strategy requires at least two relays to be employed, since one relay receives the source's frame while another relay is forwarding a previous frame to the destination, thus recovering the half-duplex loss of regular relays, as the destination receives one frame per transmission phase with the exception of the first phase. However, overlapping transmissions result in IRI and the source has to consider the interference power that the candidate relay for reception will receive by the transmitting relay. If successive transmission is infeasible, the transmission policy reduces to single link selection, where either the source or a relay transmits a packet, following a policy similar to \cite{KRICHAR}. 
\begin{figure}[h]
\vspace{-.2cm}
\centering
\includegraphics[width=0.45\columnwidth]{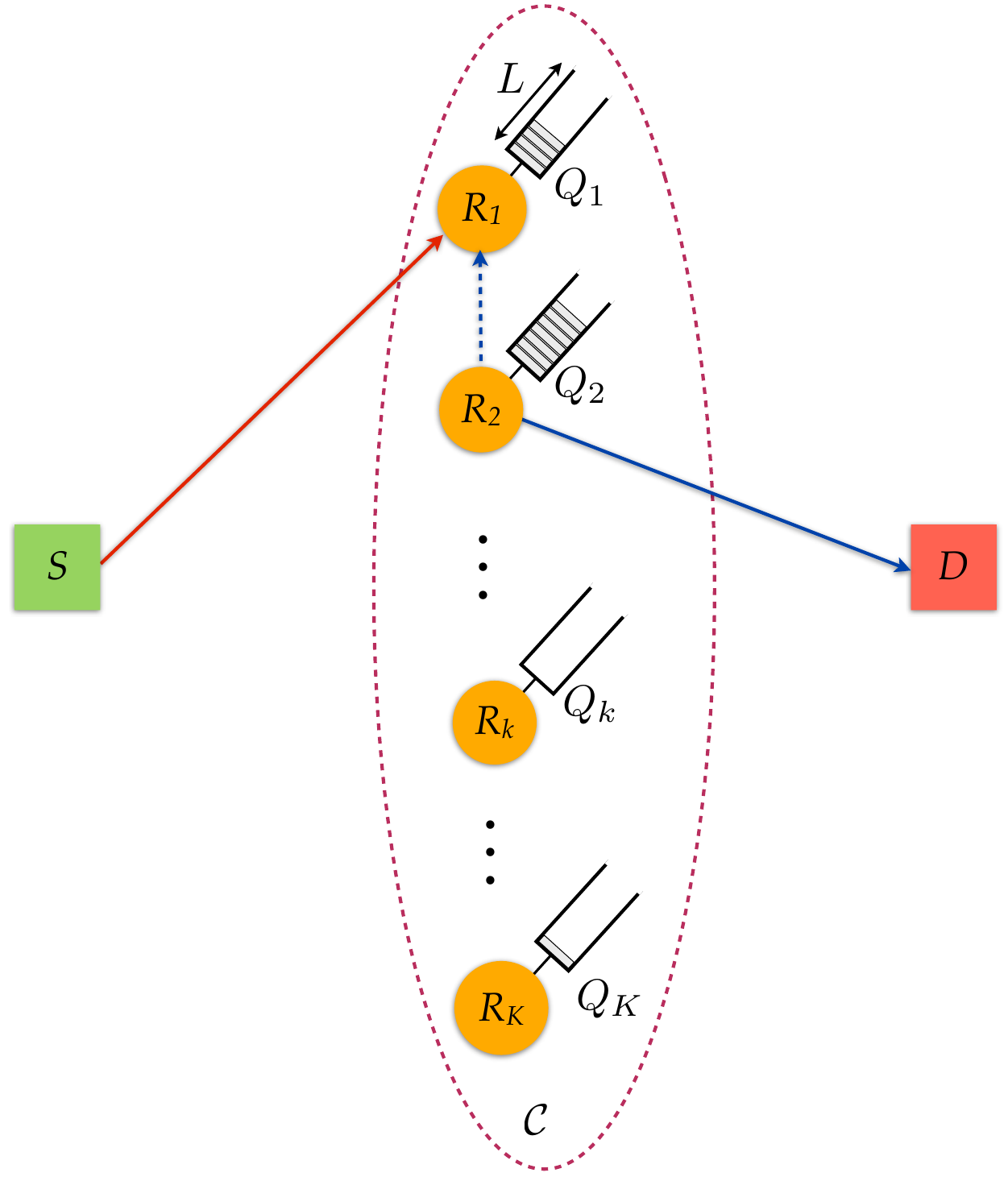}
\caption{The system model: Source $S$ communicates with Destination $D$ via a cluster of relays $R_k \in \mathcal{C}$, $k\in[1,K]$.}
\vspace{-.2cm}
\label{overflow}
\end{figure}

Successive relaying is combined with opportunistic relaying in order to select which pair of relays will assist the communication in each transmission phase. So, before the source starts transmitting a new frame, a relay pair that can fulfill a spectral-efficiency target, denoted by $r_0$, with the minimum power expenditure is chosen, since the power adaptation which is performed in accordance to the SINR requirements at each link, is equivalent to the minimization of the total power in the network subject to the SINR constraints (see, for example, \cite{FM1993, Themis}). 

The interference power at the $k^{th}$ receiver, denoted by $I_{k}$, in its general case, includes the interference from all the transmitters (belonging in a set $\mathcal{S}$ with cardinality $|S|$) in the network apart from the intended transmitter $i$ and the thermal noise, and is given by
\begin{align}\label{interference}
I_{k}(\mathbf{P})\triangleq\sum_{j\neq i, j \in \mathcal{S}}{g_{jk}P_{j} + n_{k}},
\end{align}
where $\mathbf{P}=[P_1 \ \ P_2 \ \ldots \ P_{|S|}]^T$. Therefore, the SINR at receiver $k$ is given by
\begin{align}\label{SINR1}
\Gamma_{k}(\mathbf{P})=\frac{g_{ik}P_{i}}{\sum_{j\neq i, j \in \mathcal{S}}{g_{jk}P_{j} + n_{k}}}\ .
\end{align}

In this work, a Gaussian input distribution and an information theoretic capacity achieving channel coding scheme is assumed, and as a result a target data rate can be easily expressed in terms of the SNR or SINR using Shannon's Theorem, i.e.,
\begin{align}
C=\log_2(1+\Gamma_{k}(\mathbf{P}))
\end{align}

Hence, the Quality of Service (QoS) can be measured in terms of SINR, i.e., independently of nodal distribution and traffic pattern, a transmission from transmitter to its corresponding receiver is successful (error-free) if the SINR of the receiver is greater or equal to the \textit{capture ratio} $\gamma_0$. The value of $\gamma_0$ depends on the modulation and coding characteristics of the radio, such as, the required data rate of the application which is supported by the network and the error-correction coding technique. Therefore, we require that
\begin{align}\label{SINR_geq}
\frac{g_{ik}P_{i}}{\sum_{j\neq i, j \in \mathcal{S}}{g_{jk}P_{j} + n_{k}}} \geq \gamma_0.
\end{align}
An outage occurs at the $k$-th receiver when $\Gamma_{k}(\mathbf{P}) < \gamma_0$, and we denote this outage probability by $\mathbb{P}(\Gamma_{k}(\mathbf{P}) < \gamma_0)$. This framework has been widely used (e.g., \cite{SAD}) and is equivalent to the \emph{capture model} introduced in \cite{HKL}.

Due to battery limitations, we assume that each transmitting node $i$ (source and relays) has a maximum peak power $P_{i}^{\max}$. So, even though we assume CSI to be known at the transmitter and the receiver, outage events emerge due to the power constraint $P_{i}^{\max}$ which might not be able to satisfy the target SNR value for correct signal reception. The power constraint is, in general, a critical design parameter in relay networks \cite{KGG}.

\subsection{CSI acquisition}

It is worth noting that our focus is to investigate the performance of a new buffer-aided successive opportunistic relay selection scheme under a global Channel State Information (CSI) assumption and hence, implementation issues  (i.e., distributed implementations \cite{BLE}, CSI acquisition, outdated CSI \cite{SCH_TCOM1} etc.) are beyond the scope of this work. In general, conventional centralized/distributed half-duplex relay selection protocols can be applied for the implementation of the proposed scheme (see, e.g., \cite{BLE1,NOS, SHAH}). However, some details about the realization of $\minpow$ are provided for both centralized and distributed approaches. 

First, in the centralized approach, we employ the source to select the relay-pair in each time-slot and keep track of the relays' buffer status. At the start of each transmission phase, the destination sends a training sequence to the relays and the RD CSI is acquired. Next, the relays take turns according to their indices, to inform the source about their RD channels. The use of training sequences for CSI acquisition is adopted by other related works, as is the case in \cite{IKH3}. As each relay transmits, the other $K-1$ relays listen to its transmission and determine their respective Relay-to-Relay (RR) link to it as pilot sequences are transmitted together with the RD status. The source receives these sequences and calculates the SR links' CSI. When all the relays have finished transmitting, the required information for centralized relay-pair selection has been gathered. Compared to \cite{IKH3} where IRI is considered negligible and relay selection is performed by the destination, an additional overhead comes from the acquisition of the RR links' CSI which are required in order for IRI to be canceled or avoided through the selection process. Nonetheless, regardless of whether IRI is negligible or not, each relay's RD link condition needs to be fed back to the source and as the other $K-1$ relays overhear this transmission, the RR channels are estimated. In this way, when their turn comes, RR CSI is made known to the source. As a result, in practice, our scheme does not introduce significant overhead compared to other state-of-the-art relay selection policies.

The distributed approach for the relay-pair selection process is based on the use of synchronized timers as proposed in \cite{BLE} for single relay selection. Recently, distributed joint-relay pair selection was presented in \cite{NOM_CAMAD} where the idea of Distributed Switch and Stay Combining (DSSC) \cite{DSSC1, DSSC2} was extended for the buffer-aided successive opportunistic relaying scenario. More specifically, each relay is first examined as a transmitting relay and through a pilot sequence, the other $K-1$ relays extract the RR links' CSI. As a result, each pair can determine the level of power minimization that it can  achieve and sets its timers to be inversely proportional to the required power for successful transmission. As all the relays take turns to transmit their pilot sequences, timer values are updated and at the end of this process, the best relay-pair timers have the minimum value and expire first. One advantage of the distributed implementation is that buffer size information is not required to be kept by the source, as is the case in the centralized approach, thus leading to significant processing reduction. Note that distributed relay-pair selection is further facilitated in networks where channel correlation is observed after each transmission phase. In this case, the relay-pair relationship relies on this correlation and does not change as frequently as in the case of networks where fading coefficients change after each transmission phase. For further CSI overhead reduction, DSSC approaches define a threshold which in our scheme could be the level of power minimization for a spectral efficiency target. In this way, the decision of switching a relay-pair for another would not be required in every transmission phase if this threshold is satisfied and the complexity is significantly reduced. Also, relay-pair selection could stop when the first pair which satisfies this threshold is found, so as to avoid the estimation of every link in the network. Here, the optimal pair is selected each time so, $K(K-1)$ RR links need to be estimated.

\vspace{-.2cm}

%
%
\section{The $\minpow$ relay selection policy}\label{sec:minpow}

In this section, we present a novel relay selection policy called $\minpow$. This relay selection scheme is associated with a one-slot cooperative protocol (similar to the $\maxlink$ relay selection policy \cite{KRICHAR} and the one in \cite{IKH3} where, however, IRI and power minimization are not considered), rather than two-slot protocols (as in \cite{IKH2}). At each time slot, the source $S$ attempts to transmit data to a selected relay with a non-full buffer (i.e., $R_r \in \mathcal{A}$), and at the same time another relay with a non-empty buffer (i.e., $R_t\in\mathcal{T}$, $R_t \neq R_r$) attempts to transmit data to the destination $D$.

In an arbitrary transmission phase, a packet is successfully transmitted from the transmitting relay $R_t$ to the destination $D$ if the SNR, denoted by $\text{SNR}_{R_t D}$, is greater than or equal to the capture ratio $\gamma_0$, i.e.,
\begin{align}\label{eq:SNR_RD}
\frac{g_{R_t D}P_{R_t}}{n_D} \geq  \gamma_0 \; , \quad R_t\in\mathcal{T}, R_t \neq R_r \; .
\end{align}
A packet is successfully transmitted from source  $S$ to the receiving relay $R_r$, if the Signal-to-Interference-and-Noise Ratio (SINR) at the receiving relay, denoted by $\text{SINR}_{S R_r}$ is greater than or equal to $\gamma_0$, i.e.,
\begin{align}\label{eq:SNR_SR}
\frac{g_{S R_r}P_{S}}{g_{R_t R_r} P_{R_t} \mathbb{I}(R_t , R_r)  +  n_{R_r}} \geq  \gamma_0 \; , \quad R_r\in\mathcal{C}, R_r \neq R_t \; ,
\end{align}
where $ \mathbb{I}(R_t , R_r) $ is a factor indicating whether interference cancellation is satisfied and it is described by
\begin{align}\label{eq:index}
\mathbb{I}(R_t , R_r)  = 
\begin{cases} 
0 , & \text{if } \displaystyle \frac{g_{R_t R_r} P_{R_t}}{g_{S R_r}P_{S} +  n_{R_r}} \geq  \gamma_0 \; ,   \\
1 , &\text{otherwise.}
\end{cases}
\end{align}

The following proposition states that if the maximum powers $P_{S}^{\max}$ and $P_{R_t}^{\max}$ are large enough (thus not imposing any limitations/constraints), for each pair of relays $R_r$ and $R_t$, then we can always find power levels such that interference cancellation conditions are satisfied. When interference cancellation is feasible, a fact expressed by $ \mathbb{I}(R_t , R_r) =0$, the interfering signal is firstly decoded and then subtracted at the relay prior to the decoding of the source signal. The strong interference regime was studied in \cite{SATO, COSTA} and more recently in \cite{FAN} and in this case the outage probability is not affected by the IRI.

\begin{proposition}\label{prop:1}
Let $P_{S}^{\max}=\infty$ and $P_{R_t}^{\max}=\infty$. For each pair of relays $R_r$ and $R_t$, there exist $P_{S}$ and $P_{R_t}$ such that $\mathbb{I}(R_t , R_r)  =0$, $\text{SNR}_{R_t D} \geq \gamma_0$ and  $\text{SINR}_{S R_r} \geq \gamma_0$. The minimum power levels $P_{S}^*$ and $P_{R_t}^*$ are achieved when $\text{SNR}_{R_t D} = \text{SINR}_{S R_r} = \gamma_0$, and are given by
\begin{subequations}\label{eq:P_optimal}
\begin{align}
P_{S}^* &= \frac{ \gamma_0 n_{R_r}}{g_{S R_r}} \; , \label{eq:P_S_optimal} \\
P_{R_t}^* &=\max \left\{ \frac{\gamma_0 n_D}{g_{R_t D}} , \frac{n_{R_r}\gamma_0 (\gamma_0 + 1)}{g_{R_t R_r}} \right\} \; \label{eq:P_Rt_optimal}.
\end{align}
\end{subequations}
\end{proposition}

\begin{proof}
For IC to take place, according to \eqref{eq:index}, we have
\begin{align}\label{eq:index2}
{g_{R_t R_r} P_{R_t}} \geq \gamma_0({g_{S R_r}P_{S} +  n_{R_r}})
\end{align}
Given that $P_{R_t}$ is chosen such that \eqref{eq:index2} is fulfilled, then \eqref{eq:SNR_SR} becomes
$$
\frac{g_{S R_r}P_{S}}{n_{R_r}} \geq  \gamma_0\; ,
$$
and since $P_S$ decreases monotonically with $P_{R_t}$ (see equation~\eqref{eq:index} for $\mathbb{I}(R_t , R_r)  =0$), while it requires a minimum value such that equation~\eqref{eq:SNR_SR} is fulfilled, the minimum power of $S$ is given with equality; i.e.,
\begin{align}\label{eq:P_S_optimal1}
P_{S}^* = \frac{ \gamma_0 n_{R_r}}{g_{S R_r}} \; .
\end{align}
Substituting \eqref{eq:P_S_optimal1} into \eqref{eq:index2} we have
${g_{R_t R_r} P_{R_t}} \geq n_{R_r}\gamma_0 (\gamma_0 + 1)$.
Hence, 
the minimum $P_{R_t}$ is 
\begin{align*}
P_{R_t}^* =\max \left\{ \frac{\gamma_0 n_D}{g_{R_t D}} , \frac{n_{R_r}\gamma_0 (\gamma_0 + 1)}{g_{R_t R_r}} \right\} \; .
\end{align*}
\end{proof}

Proposition \ref{prop:1} provides the minimum power levels of $S$ and $R_t$, provided that their maximum power levels do not impose any constraint, and hence, the IC conditions are satisfied. In the next proposition, we find the optimal power levels of $S$ and $R_t$ when IC cannot take place.

\begin{proposition}\label{prop:2}
For each pair of relays $R_r$ and $R_t$, when interference cancellation is infeasible, the signal from $R_t$ can be decoded successfully at $D$ if and only if
\begin{align}
\frac{\gamma_0 n_D}{g_{R_t D}}&\leq  P_{R_t}^{\max} \; .\label{subeq:ineq1} 
\end{align}
In addition, the signal from $S$ can be decoded successfully at $R_r$ if and only if
\begin{align}
\gamma_0 \left(\frac{g_{R_t R_r}}{g_{S R_r}} P_{R_t}  +  \frac{n_{R_r}}{g_{S R_r}}\right) &\leq P_{S}^{\max} \; . \label{subeq:ineq2}
\end{align}
When \eqref{subeq:ineq1} and \eqref{subeq:ineq2} hold, the minimum power levels $P_{S}^\dagger $ and $P_{R_t}^\dagger $ are achieved when $\text{SNR}_{R_t D} = \text{SINR}_{S R_r} = \gamma_0$, and are given by
\begin{subequations}\label{eq:P_optimal2}
\begin{align}
P_{S}^\dagger &= \gamma_0 \left(\frac{g_{R_t R_r}}{g_{S R_r}} \frac{\gamma_0 n_D}{g_{R_t D}}  +  \frac{n_{R_r}}{g_{S R_r}}\right) \; , \label{eq:P_S_optimal2} \\
P_{R_t}^\dagger &=\frac{\gamma_0 n_D}{g_{R_t D}} \; \label{eq:P_Rt_optimal2}.
\end{align}
\end{subequations}
\end{proposition}

\begin{proof}
Interference cancellation cannot take place when the maximum power of $R_t$ is not high enough, such that its signal can be decoded by $R_r$, i.e.,
\begin{align}\label{eq:P_R_t_max}
 \frac{g_{R_t R_r} P_{R_t}^{\max}}{g_{S R_r}P_{S} +  n_{R_r}} <  \gamma_0 \; ,
\end{align}
but it should definitely be high enough to be decoded at the destination, i.e.,
\begin{align}
\frac{\gamma_0 n_D}{g_{R_t D}}\leq P_{R_t} &\leq  P_{R_t}^{\max} \label{subeq:ineq11}.
\end{align}
Then,  $P_S$ should be high enough, so that the transmitting signal can be successfully decoded at $R_r$. Hence,
\begin{align}
\gamma_0 \left(\frac{g_{R_t R_r}}{g_{S R_r}} P_{R_t}  +  \frac{n_{R_r}}{g_{S R_r}}\right) \leq P_{S} &\leq P_{S}^{\max}\;. \label{subeq:ineq21}
\end{align}
Substituting \eqref{subeq:ineq11} into \eqref{subeq:ineq21} we have 
\begin{align} \label{eq:condition}
\gamma_0 \left(\frac{g_{R_t R_r}}{g_{S R_r}} \frac{\gamma_0 n_D}{g_{R_t D}} +  \frac{n_{R_r}}{g_{S R_r}}\right) \leq P_{S} &\leq P_{S}^{\max} \; .
\end{align}
\end{proof}

\begin{coro}\label{feasibility check}
From equation~\eqref{eq:P_Rt_optimal} in Proposition~\ref{prop:1} it is obvious that IC can take place if and only if 
\begin{align}\label{eq:Pmax_condition}
P_{R_t}^{\max} \geq \max \left\{ \frac{\gamma_0 n_D}{g_{R_t D}} , \frac{n_{R_r}\gamma_0 (\gamma_0 + 1)}{g_{R_t R_r}} \right\} \;.
\end{align}
Hence, someone can check whether IC is feasible by checking the validity of inequality~\eqref{eq:Pmax_condition}. Hereafter, this is called the {IC feasibility check}.
\end{coro}

Making use of Propositions \ref{prop:1} and \ref{prop:2} and Corollary~\ref{feasibility check}, we now describe the $\minpow$ relay selection algorithm. It is assumed that at each time step the source collects the CSI information required. We denote by $\mathcal{P}$ the set of all possible relay-pairs in the relay network, and by $|\mathcal{P}|$ its cardinality. A relay pair, denoted by $(R_r,R_t)$ belongs to set $\mathcal{P}$ if and only if $r\neq t$, the receiving relay $R_r$ is not full and the transmitting relay $R_t$ is not empty (i.e., $R_r\in \mathcal{A}$ and $R_t\in \mathcal{T}$). 

The successive transmission part of the $\minpow$ algorithm is described as follows: \\
1. First, \emph{for each possible pair of relays} $(R_r,R_t)$, we carry out an \emph{IC feasibility check}, i.e., we check through \eqref{eq:Pmax_condition} if interference cancellation is feasible. \\
2. If IC is  \emph{feasible}, then
\begin{itemize}
\item[(i)] assuming IC did not take place, $P_{R_t}^\dagger $ and  $P_{S}^\dagger$ are as given in \eqref{eq:P_Rt_optimal2} and \eqref{eq:P_S_optimal2}, respectively. $P_{R_t}^\dagger \leq P_{R_t}^{\max}$ is satisfied by the feasibility check, but we need again to check separately whether $P_{S}^\dagger \leq P_{S}^{\max}$; 
\item[(ii)]  assuming IC took place, $P_{R_t}^*$ and  $P_{S}^*$ are as given in \eqref{eq:P_Rt_optimal} and \eqref{eq:P_S_optimal}, respectively. Note that since the feasibility criterion is satisfied, then $P_{R_t}^* \leq P_{R_t}^{\max}$. Nevertheless, we need to check separately whether $P_{S}^* \leq P_{S}^{\max}$; if it is not satisfied, then step 2(i) cannot be satisfied either. 
\item[(iii)] The minimum energy expenditure at a specific time slot for each pair is the minimum sum of the powers for the two cases, i.e., 
$\min \left\{ P_{S}^*+ P_{R_t}^*, P_{S}^\dagger + P_{R_t}^\dagger \right\}$ or any of the cases (i) and (ii) if the other case is not possible. Note that if IC can take place, this might require $P_{R_t}^* > P_{R_t}^\dagger$, such that $P_{S}^*+ P_{R_t}^* > P_{S}^\dagger + P_{R_t}^\dagger$. 
\end{itemize}
3. If IC is \emph{infeasible}, then we check whether $P_{R_t}^\dagger \leq P_{R_t}^{\max}$ and $P_{S}^\dagger \leq P_{S}^{\max}$ (as in step 2(i)). \\
4. We compare the minimum energy expenditure for all possible relay pairs $(R_r,R_t)\in \mathcal{P}$ and we choose the minimum among them. Note that, if there exist no relay pairs which can perform successive relaying, the proposed hybrid scheme selects a single link based on the $\maxlink$ approach.

\begin{remark}
Note that, while the single link relaying scheme is similar to $\maxlink$, our scheme includes power adaptation enabling the minimization of the power used for transmission. More specifically, the power adaptation which is performed in accordance to the SINR requirements at each link, is equivalent to the minimization of the total power in the network subject to the SINR constraints (see, for example, \cite{FM1993, Themis}). As a result, the selection of the best pair minimizes the energy spent to satisfy the SINR constraint.
\end{remark}
\begin{remark}
Note that in the worst case scenario (in which all the queues are neither empty nor full), there will be $|\mathcal{P}|=K\times (K-1)$ combinations. Hence,  the worst case complexity of the problem is $\mathcal{O}(K^2)$.
\end{remark}
\begin{remark}
We note that in the case of independent non-identical distributed (i.n.i.d.) channels, our scheme may experience difficulties due to the two-hop asymmetry. However, this holds for all the buffer-aided schemes that we compare our scheme with, as relay buffers will tend to get full or empty faster, depending on the hop which has the the best channel conditions. This case could be efficiently handled if transmission rate is kept fixed for both hops as is the case here, or by prioritizing the selection of the pair which achieves the more balanced two-hop transmission rates, thus leading to buffer stabilization in the long-term. In this way, relays will be available to be selected either for transmission or reception and diversity will be assured.
\end{remark}

\vspace{-.2cm}

%
%
\section{Outage probability model and analysis}\label{sec:outage}

In this section, the outage probability behavior of the $\minpow$ relay selection scheme follows the theoretical framework of \cite{KRICHAR}, which is also a relay network with finite buffers.  The main differentiation compared to \cite{KRICHAR} is that we have additional ways of transmission through successive relaying; in other words, we have additional links from a certain buffer state to others. Note that the possibility of having a successive transmission requires two links to offer a SINR at the receivers equal to or above $\gamma_0$ at the same time, otherwise transmission is based on single-link selection. So, an outage takes place when $\gamma_0$ can not be achieved even in a single-link transmission; thus, $\mathbb{P}_{\rm out}\triangleq \mathbb{P}(\Gamma_{k}(\mathbf{P}) < \gamma_0)$.

\subsection{Construction of the state transition matrix of the Markov Chain (MC)}\label{sec:outageA}
 
We first formulate the state transition matrix of the MC, denoted as $\mathbf{A}$, $\mathbf{A}\in \mathbb{R}^{{\left( {L + 1} \right)^K} \times {\left( {L + 1} \right)^K}}$. More specifically, 
${{\bf{A}}_{i,j}} = \mathbb{P}\left( {{s_i} \to {s_j}} \right) = \mathbb{P}\left( {{X_{t + 1}} = {s_j}|{X_t} = {s_i}} \right)$
are the transition probabilities to move from a buffer state $s_i$ to a state $s_j$. The transition probability depends on the number of relays that are available for cooperation. 

\begin{remark}
As we consider finite buffers, relays that have full buffers cannot compete in the selection of the best relay that will receive the source's signal. Also, relays with empty buffers are not able to transmit and as a result they are excluded from the best transmitting relay selection. Moreover, when there is no possibility of transmitting successively through two selected relays our system reduces to the $\maxlink$ relay selection scheme. The number of links that are available in this mode is reduced if the relays have either full or empty buffers.
\end{remark}

In what follows, we distinguish the outage events at each communication link. The outage event $A$ denotes the case of experiencing an outage in the $SR$ link when either IC is possible or not. Let $A_1$ denote the event that IC is impossible at relay $R_i$,
\begin{align*}
A_1=\left\{ \frac{g_{S R_i}P_S}{g_{R_t R_i}P_{R_t}+n_{R_i}}<\gamma_0 \right \}\;.
\end{align*}
When IC is possible then the event is denoted by $A_2$ and it is given by
\begin{align*}
A_2=\left\{ \frac{g_{S R_i}P_S}{n_{R_i}}<\gamma_0 \right \}\;.
\end{align*}
Events $A_1$ and $A_2$ are mutually exclusive in the sense that the transmission that takes place has either IC (event $A_2$) or not (event $A_1$);  hence, $A=A_1\cup A_2$ and $A_1\cap A_2 = \emptyset$. Equivalently, for the $RD$ link, an outage event occurs when
\begin{align*}
B=\left\{\frac{g_{R_t D}P_{R_t}}{n_{D}}<\gamma_0 \right \} \; .
\end{align*}
\begin{remark}\label{remark:exp}
Since the channel gains ($g_{ij} = |h_{ij}|^2$) are exponentially distributed, the event that the instantaneous Signal-to-Noise Ratio (SNR) is smaller than the capture ratio (i.e., $SNR< \gamma_0$) is exponentially distributed for fixed power levels (see, for example, \cite{SAD}). As a result, an outage event on the $RD$ link -- given that the relays are power-limited (i.e., $P_{R_t}\leq P_{R_t}^{\max}$) -- is exponentially distributed. Also, if there is no relay transmitting to the destination, there is no interference and an outage event on the $SR$ link is also exponentially distributed, using similar arguments.
\end{remark}
In this simplifying case for which all links are i.i.d. and symmetric, we are in outage when neither successive nor non-successive communication can take place. This is equivalent to saying that all possible links are in outage. This can be simply expressed as
\begin{align}
{\bar p}_{ij} 
 =  p(A)^{|\mathcal{A}_i|} p(B)^{|\mathcal{T}_i|} \; .
\end{align}
Despite the fact that we assume symmetric i.i.d. links, since the sets for events $A$ and $B$ may be different, we did not use the same notation for $p(A)$ and $p(B)$. Nevertheless, by letting $p \triangleq p(A)=p(B)$ (since the links are symmetric i.i.d.) and by assuming all relays are neither full nor empty (i.e., $|\mathcal{T}_i|=|\mathcal{A}_i|=K$), the outage probability is given by ${\bar p}_{ij} =p^{2K}$, i.e., the probability that all $2K$ possible links are in outage.
\begin{remark}
The outage probability of $\minpow$ is bounded by the outage probability of single-link transmission \cite{KRICHAR}, since the probability of having an outage when searching for a successive transmission is larger due to the requirement of simultaneously not having two links in outage. On the contrary, in single-link transmission, the same links are examined and if one of them fulfils the threshold for error-free transmission then the outage event is avoided.
\end{remark}

\subsection{Steady state distribution and outage probability}\label{steady_st}

Since we have defined the entries of the transition matrix $\mathbf{A}$, the next step is to find the steady state distribution of the MC. In this way, the relationship among the different ways of leaving and reaching specific buffer states will be defined.

\begin{proposition}\label{prop:3}
The state transition matrix $\mathbf{A}$ is Stochastic, Indecomposable\footnote{A stochastic matrix $P\in \mathbb{R}^{m\times m}$ is said to be decomposable if there exists a nonempty proper subset $\mathcal{S} \subset \{v_1,v_2, \ldots, v_m \}$ such that $p_{ji}=p_{ij}=0$ whenever $v_i\in \mathcal{S}$ and $v_j \notin \mathcal{S}$. $P$ is indecomposable if it is not decomposable.} and Aperiodic\footnote{In a finite state Markov Chain, a state $i$ is aperiodic if there exists $k$ such that for all $k' \geq k$, the probability of being at state $i$ after $k'$ steps is greater than zero; otherwise, the state is said to be periodic. A stochastic matrix is aperiodic if every state of the Markov chain it describes is aperiodic.} (SIA).
\end{proposition}

\begin{proof}
In order to prove that $\mathbf{A}$ is SIA, we have to show that it is (i) row stochastic, (ii) indecomposable, and (iii) aperiodic. \\
\noindent (i) \emph{Row Stochasticity}: For any MC the transition from state $s_i$ to a state $s_j$ for all possible states $s_i$ sums up to $1$, i.e.,
\begin{equation}\label{eq:row_stoch}
\sum\limits_{i = 1}^{{{\left( {L + 1} \right)}^K}} {{{\bf{A}}_{i,j}}}  = 1.
\end{equation}
In the MC that models the buffer states of the relays, if there exists a transition from state $s_i$ to $s_j$ then there exists a transition from $s_j$ to $s_i$. This fact applies to our scheme where  transitions due to successive relaying, offer additional links from one state to another and vice versa. However, since the states are not symmetric and the number of links to other states is not the same, the transition probabilities are not the same. Thus, the transition matrix $\mathbf{A}$ is not symmetric. As a result the $\mathbf{A}$ is row stochastic, but not necessarily doubly stochastic. \\
\noindent (ii) \emph{Indecomposability}: Due to the structure of the problem all the possible states of the considered MC can communicate and hence its state space is a single communicating class; in other words, it is possible to get to any state from any state. Hence, the MC is indecomposable. \\
\noindent (iii) \emph{Aperiodicity}: Aperiodicity of a MC is easily established due to the fact that the diagonal entries that correspond to original (non-virtual) nodes are nonzero. All links receive nonnegative weights and the diagonals (outage probabilities) are strictly positive. Also, the probability of being at any state after $M$ and $M +1$ transitions is greater than zero; hence, all states are aperiodic and therefore, the MC is aperiodic.
\end{proof} 

\begin{lemma}\label{lemma:1}
The stationary distribution of the row stochastic matrix \textnormal{\textbf{A}} of the MC that models the buffer states is given by
$\bm{\pi}=(\mathbf{A}-\mathbf{I}+\mathbf{B})^{-1}\bf{b}$, where $\bm{\pi}$ is the stationary distribution, $\bf{b}=(1~1~\ldots~1)^T$ and $\mathbf{B}_{i,j}=1, ~\forall i,j$.
\end{lemma}

\begin{proof}
Similar to that of \cite[Lemma~1, Lemma~2]{KRICHAR}.
\end{proof}

\subsection{Derivation of the outage probability}

We assume that we have outage when both the $SR$ and $RD$ links are in outage. In this case, no packet is moved around and hence the system remains at the same state. Using the steady state of the MC and the fact that an outage event occurs when there is no change in the buffer
status\footnote{Note that $\minpow$ first checks if a successive transmission is possible and if it fails due to one or both hops being in outage, then it operates as a single-link selection scheme. If single-link selection fails then there is no change in the buffer status and the system is in outage.}, the outage probability of the system can be expressed as \cite{KRICHAR},
\begin{equation}\label{eq:p_out}
{P_{out}} = \sum\limits_{i = 1}^{{{\left( {L + 1} \right)}^K}} {{{\bm{\pi }}_i}{{\bar p}_{ij}} = {\rm{diag}}\left( {\bf{A}} \right){\bm{\pi }}} \; .
\end{equation}
By constructing the state transition matrix \textbf{A} that captures in its diagonal the probabilities where no change in buffer states happened, and the corresponding steady state probabilities, we can easily compute the outage probability of the system. 

\vspace{-.2cm}

%
%
\section{Illustrative Examples}\label{sec:examples}

In the previous section, we have described the theoretical framework for the computation of the outage probability. In what follows, we will present two illustrative examples that showcase the behavior of our approach for different parameters. The first example consists of two relays $(K=2)$ with finite buffer size equal to two $(L=2)$, while the second one examines the case of infinite buffers \((L\rightarrow\infty)\) at the relays.

\subsection{Illustrative example of K=2 relays with buffer size L=2}\label{sec:examples1}

This illustrative example showcases the behavior of our approach for different parameters. Since we have a scheme that employs successive transmissions the simplest case is when two relays are available. Assuming that each relay has a buffer size equal to two, we show its state transition diagram in Fig.~\ref{overflow}, with the nine possible states for the buffers of the two relays.

\begin{figure}[ht]
\vspace{-.2cm}
\centering
\includegraphics[width=0.65\columnwidth]{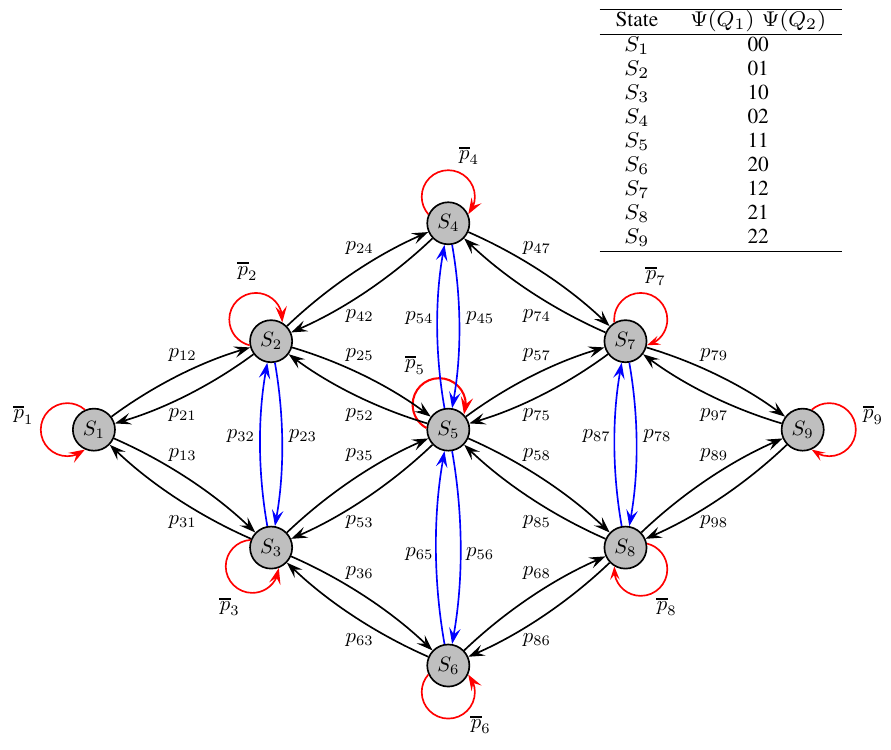}
\vspace{-0.3cm}
\caption{State diagram of the Markov chain representing the states of the buffers and the transitions between them for a case with $K=2
$ and $L=2$. Compared to the $\maxlink$ scheme in \cite{KRICHAR}, the $\minpow$ model includes extra transition states due to the successive nature of the protocol.}
\vspace{-.2cm}
\label{overflow}
\end{figure}

The steady state of the system for different values of SNR can be found by using the method described in \ref{steady_st} and the corresponding state transition matrix \textbf{A} is as follows.
\begin{align*}
{\small
\mathbf{A} =  \left(\begin{array}{ccccccccc} 
   \bar{p}_{1} & p_{12} & p_{13} & {0}        & {0}        & {0}        & {0}        & {0}        & {0}  \\
   p_{21} & \bar{p}_{2} & p_{23} & p_{24} & p_{25} & {0}        & {0}        & {0}        & {0}  \\
   p_{31} & p_{32} & \bar{p}_{3} & {0}        & p_{35} & p_{36} & {0}        & {0}        & {0}  \\
   0           & p_{42} & {0}        & \bar{p}_{4} & p_{45} & {0}        & p_{47} & {0}        & {0}  \\
   0           & p_{52} & p_{53} & p_{54} & \bar{p}_{5} & p_{56} & p_{57} & p_{58} & {0}  \\
   0           & {0}        & p_{63} & {0}        & p_{65} & \bar{p}_{6} & {0}        & p_{66} & {0}  \\
   0           & {0}        & {0}        & p_{74} & p_{75} & {0}        & \bar{p}_{7} & p_{78} & p_{79}  \\
   0           & {0}        & {0}        & {0}        & p_{85} & p_{86} & p_{87} & \bar{p}_{8} & p_{89}  \\
   0           & {0}        & {0}        & {0}        & {0}        & {0}        & p_{97} & p_{98} & \bar{p}_{9} 
\end{array} \right) \; .}
\end{align*}

\begin{table}[h!]
\centering
\caption{Buffer states for $K=2$ relays and $L=2$ buffer size}
{\small
\begin{tabular}{c c c c c}
\hline
Buffer State & $\bm{\pi}$ ($P=0$dB) & $\bm{\pi}$ ($P=10$dB) & $\bm{\pi}$ ($P=20$dB) & $\bm{\pi}$ ($P=30$dB) \\
\hline
$S_1$ & 0.0413774 & 0.0005918 & 0.0000055 & 0\\
$S_2$ & 0.1264042 & 0.1021872 & 0.1001714 & 0.05115971\\
$S_3$ & 0.1263465 & 0.1021954 & 0.1001701 & 0.05115971\\
$S_4$ & 0.1145455 & 0.1481283 & 0.1498316 & 0.19884035\\
$S_5$ & 0.1826525 & 0.2937944 & 0.2996428 & 0.39768046\\
$S_6$ & 0.1145455 & 0.1481283 & 0.1498316 & 0.19884035\\
$S_7$ & 0.1263465 & 0.1021954 & 0.1001701 & 0.05115971\\
$S_8$ & 0.1264042 & 0.1021872 & 0.1001714 & 0.05115971\\
$S_9$ & 0.0413774 & 0.0005918 & 0.0000055 & 0\\
\hline
\end{tabular}}
\label{K2L2_pi}
\end{table}

In Table~\ref{K2L2_pi} we observe that when SNR increases, the steady state distribution decreases in buffer states $S_1$ and $S_9$, and eventually it becomes practically zero. That means that for high SNR, the probability of outage tends to zero. On the other hand, when SNR increases, the steady state distribution in state $S_5$ increases, which dominates the states for large SNRs.

In the numerical results (Section~\ref{sec:numerical}), we evaluate the theoretical framework for the case of $K=2$, $L=2$ in order to examine its  behavior compared to obtained numerical results from simulations.

\subsection{Example of buffer size $L\rightarrow\infty$ and transmission power $P\rightarrow\infty$}\label{sec:examples2}

The asymptotic analysis yields interesting results due to the nature of the problem. In \cite{KRICHAR} it was shown that as the buffer size approaches infinity, the states where no full or empty relays exist are dominant\footnote{By dominant we mean the states for which the probability of being at that state, when steady state is reached, is greater than zero.}. 
As described in Remark \ref{remark:exp} of Section~\ref{sec:outageA}, the outage probability of the network is exponentially distributed and is written analytically as
\begin{align}\label{eq:Pout}
{\bar p_{ij}} = {\left( {1 - \exp \left( { - \frac{\gamma_0}{P}} \right)} \right)^{2K}} \; .
\end{align}
Below the derivation of the diversity order of the proposed $\minpow$ scheme is presented.
\begin{lemma}\label{lemma:3}
The diversity order\footnote{The diversity order (or diversity gain), denoted herein by $d$, is the gain in spatial diversity, used to improve the reliability of a link and it is defined as follows: $d = -\lim_{\textrm{SNR} \rightarrow \infty} \frac{\log \mathbb{P}_{\rm out}(\textrm{SNR})}{\log \textrm{SNR}}$. } of $\minpow$ is equal to $2K$. 
\end{lemma}

\begin{proof}
The diversity order is derived using its definition and equation \eqref{eq:Pout}, i.e.,  
\begin{align*}
d &= -\lim_{P \rightarrow \infty} \frac{\log \mathbb{P}_{\rm out}(P)}{\log P} 
=-\lim_{P \rightarrow \infty} \frac{\log {\left( {1 - \exp \left( { - \frac{{\gamma_0}}{P}} \right)} \right)^{2K}}}{\log P}
\stackrel{(a)}{\approx}-\lim_{P \rightarrow \infty} \frac{\log {\left(\frac{\gamma_0}{P} \right)^{2K}}}{\log P} \\
&=-\lim_{P \rightarrow \infty} \frac{2K\log\left( \frac{\gamma_0}{P} \right)}{\log P} =-2K\lim_{P \rightarrow \infty} \frac{\log \left(\gamma_0\right)}{\log P} + 2K\lim_{P \rightarrow \infty}  \frac{\log P}{\log P}
=2K \; .
\end{align*}
Note that approximation $(a)$ emerges from the fact that if $x \rightarrow 0$, then $1-e^{-x} \approx x$. 
\end{proof}

For finite buffer sizes, possible gains in outage probability compared to \cite{KRICHAR} derive from better interconnection between the buffer states. Buffer states where relays are full or empty are often avoided and the states where increased diversity is offered, are more usual.

\vspace{-.2cm}

%
%
\section{Numerical Results}\label{sec:numerical}

In line with the previous discussion, we have developed a simulation setup for the $\minpow$ scheme to evaluate its performance with a spectral efficiency target $r_0=1$ bps/Hz, in terms of: $1)$ outage probability, $2)$ average throughput, $3)$ power reduction and $4)$ average delay. The $\minpow$ scheme proposed herein is compared to best-relay selection (BRS) \cite{BLE}, successive opportunistic relaying (SOR) \cite{NOM}, hybrid relay selection ($\max-\max$) \cite {IKH2} and $\maxlink$ selection \cite{KRICHAR}. In addition, we provide a Selection Bound corresponding to the case presented in \cite{IKH3} where inter-relay interference is ignored and additionally all links are always  available for selection, i.e., buffers are neither full nor empty. Also, the Selection Bound scheme is coupled with single-link transmissions when successive transmissions fail, in order to provide a fair comparison with $\minpow$. As a result, the Selection Bound is an upper bound not only for $\minpow$ which allows successive transmissions, but also for the rest of the schemes which are included in the comparisons, due to its single-link selection capability.

\subsection{Outage Probability}

Fig.~\ref{pout_K2L2}, illustrates the outage probability results. Each scheme employs $K=2$ relays with buffer size $L=2$. The selection policy that offers the worst performance is SOR. The lack of buffers prohibits the combination with a more robust scheme, such as the $\maxlink$ and IRI degrades the outage performance. $\max-\min$ shows better behavior as IRI is not present. Furthermore, $\max-\max$ offers about 1.5 dB improvement due to the use of buffers, over BRS. Even better results are achieved by $\maxlink$ as outage performance is improved by almost 4 dB due to the flexibility in the link selection. As $\maxlink$ is a part of $\minpow$ when successive transmissions are not possible, we observe similar results between these two schemes. For $K=2$ and $L=2$, $\minpow$ exhibits a 0.5 dB gain for high SNR. The increased interconnection between buffer states guarantees that states $S_{1}$ (00) and $S_{9}$ (22) offering the least diversity, are more often avoided, compared to the $\maxlink$ scheme. Also, the theoretical curve of the outage probability matches the simulation results validating the analysis in Section~\ref{sec:outage}.

\begin{figure}[h]
\vspace{-0.4cm}
\centering
\includegraphics[width=8cm,height=6.7cm]{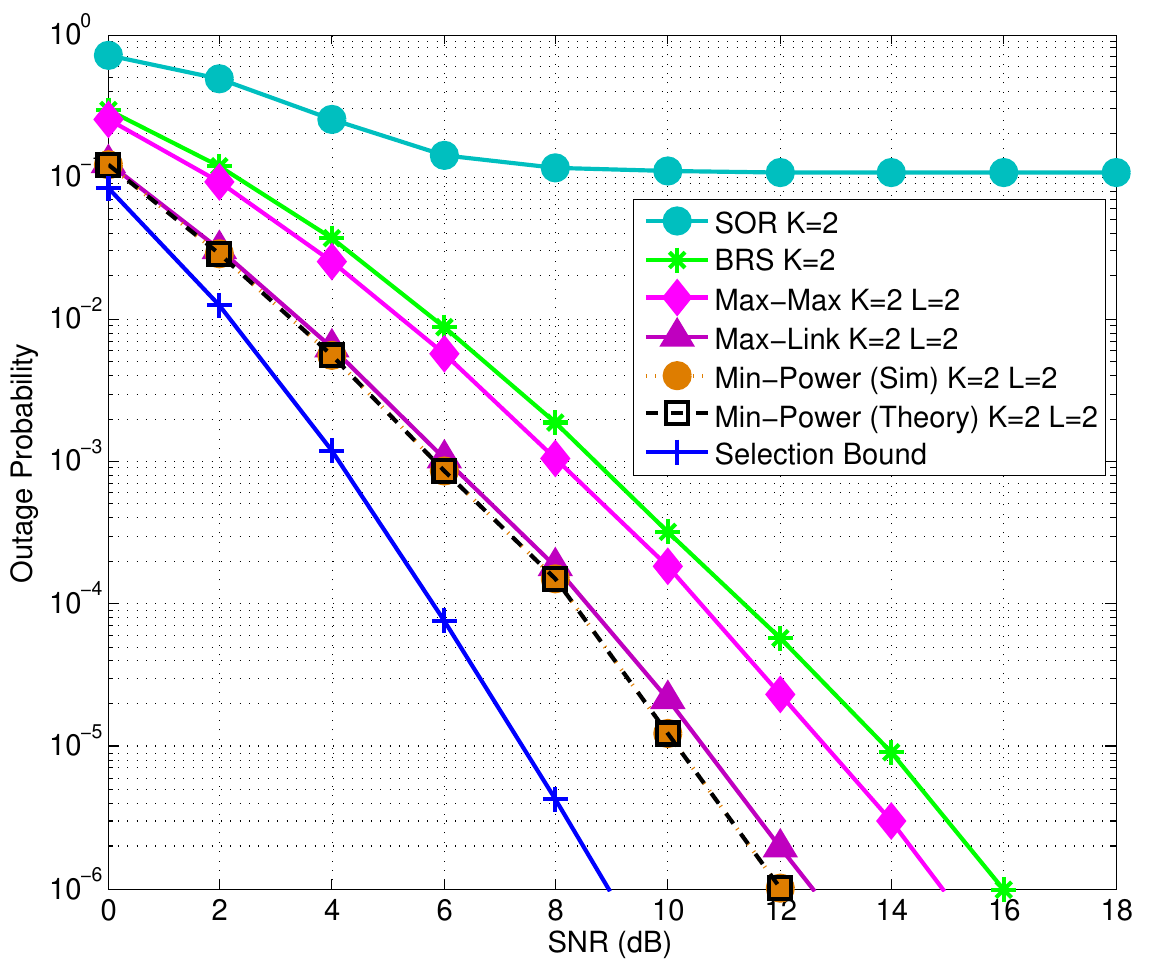}
\vspace{-0.45cm}
\caption{Outage probability for increasing transmit SNR for $K=L=2$. A 0.5 dB gain is introduced by $\minpow$  over $\maxlink$ due to increased interconnection between buffer states.}
\vspace{-.6cm}
\label{pout_K2L2}
\end{figure}


In Fig.~\ref{pout_K2Lv}, we depict the outage probability performance for increasing transmit SNR and varying $L$. As $L$ increases, the curves become steeper, thus indicating the increase in diversity as more links are available for selection. For $L=100$ and $L=\infty$, the outage curves almost match but still retain a gap from the Selection Bound case. This is reasonable as IRI degrades the performance of $\minpow$, while for the Selection Bound we assumed that the relays are isolated and inter-relay interference is negligible. Still, the $\minpow$ curve follows closely the Selection Bound especially for high SNR as very strong interference increases the probability of interference cancellation.
\begin{figure}[h]
\vspace{-0.4cm}
\centering
\includegraphics[width=8cm,height=6.7cm]{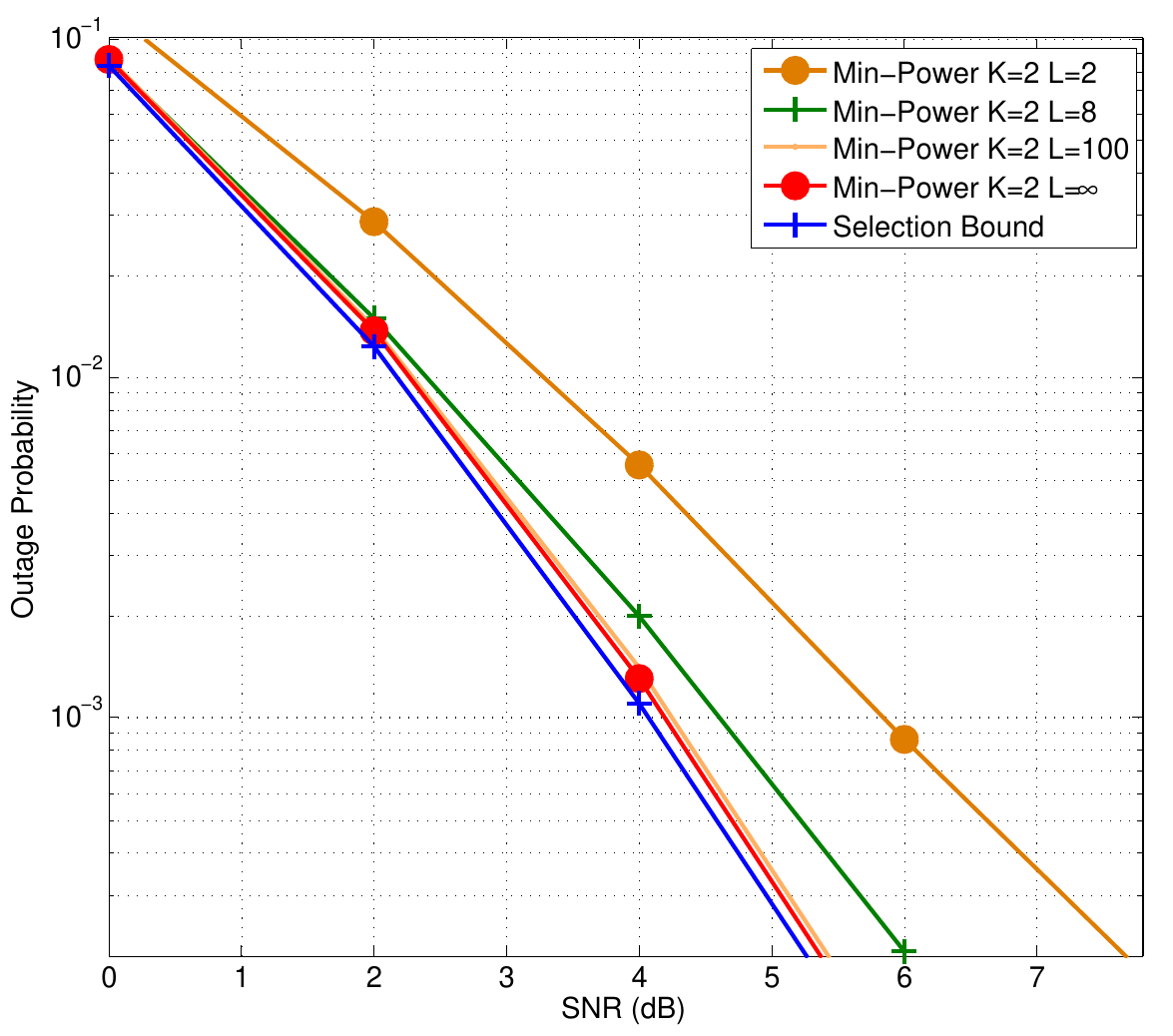}
\vspace{-0.45cm}
\caption{Outage probability for increasing transmit SNR for $K=2$ and varying $L$. For increased $L$, the curves become steeper thus indicating increased diversity as more links are available for selection.}
\vspace{-.6cm}
\label{pout_K2Lv}
\end{figure}


Next, Fig.~\ref{pout_KvL4}, shows the results for varying $K$ while $L=4$. As we saw in the asymptotic analysis example, diversity increases with a rate twice the number of $K$. Here we have a fixed finite buffer size but still each relay addition obviously improves the achieved diversity of the network. This derives from the fact that the possibility of inter-relay interference cancellation increases with an order equal to $K(K-1)$ as we have more relay pairs to select from. It is interesting to note, whenever $K$ increases to three from two the improvement is bigger than the gain introduced by adding one more relay when $K=3$ as $r_0$ is fulfilled even for low $K$ values.
\begin{figure}[h]
\vspace{-0.4cm}
\centering
\includegraphics[width=8cm,height=6.7cm]{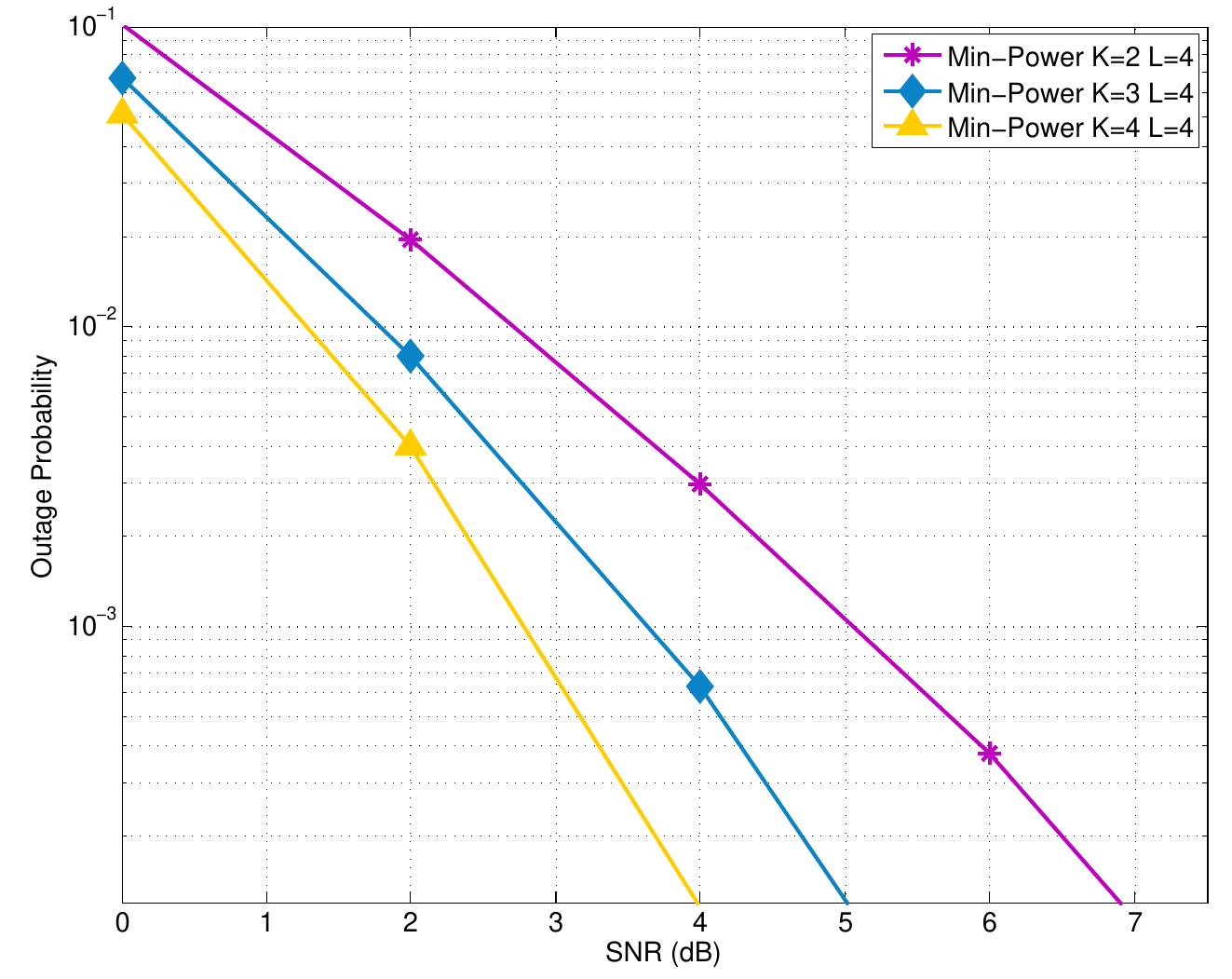}
\vspace{-0.45cm}
\caption{Outage probability for increasing transmit SNR for varying $K$ and $L=4$. Each relay addition improves  the diversity of the network as we have more relay pairs to select from, thus increasing the chances for IRI cancellation.}
\vspace{-.6cm}
\label{pout_KvL4}
\end{figure}


\subsection{Average Throughput}

For the second set of comparisons among relay selection schemes, we present in Fig.~\ref{thr_K2L2} the average throughput performance, measured in bps/Hz.   First, we see that the compared selection policies are divided in two groups. The first one consists of the half-duplex schemes, namely BRS, $\max-\max$ and $\maxlink$. Due to the constant transmission rate, equal to 1 bps/Hz, these schemes can achieve a maximum average throughput of 0.5 bps/Hz. In line with the outage probability performance, $\maxlink$ outperforms BRS and $\max-\max$ and reaches the upper bound nearly 2.5 dB prior to the others. In the second group we have SOR and $\minpow$. These schemes aim to lift the half-duplex constraint and increase the average throughput through successive transmissions. It is observed that the schemes of the second group reach the upper-limit more slowly than those of the first group. This is explained by the fact that in the low SNR regime, single-link transmissions are often performed, thus reducing the throughput by one-half in these cases. Also, in $\minpow$'s case, IRI degrades its performance and thus, introduces a delay in reaching its upper bound contrary to the IRI-free HD schemes. Min-power achieves the best performance reaching 1 bps/Hz for high SNR. SOR however, does not reach the upper bound even for high SNR as IRI causes many outages. 
\begin{figure}[h]
\vspace{-0.4cm}
\centering
\includegraphics[width=8cm,height=6.7cm]{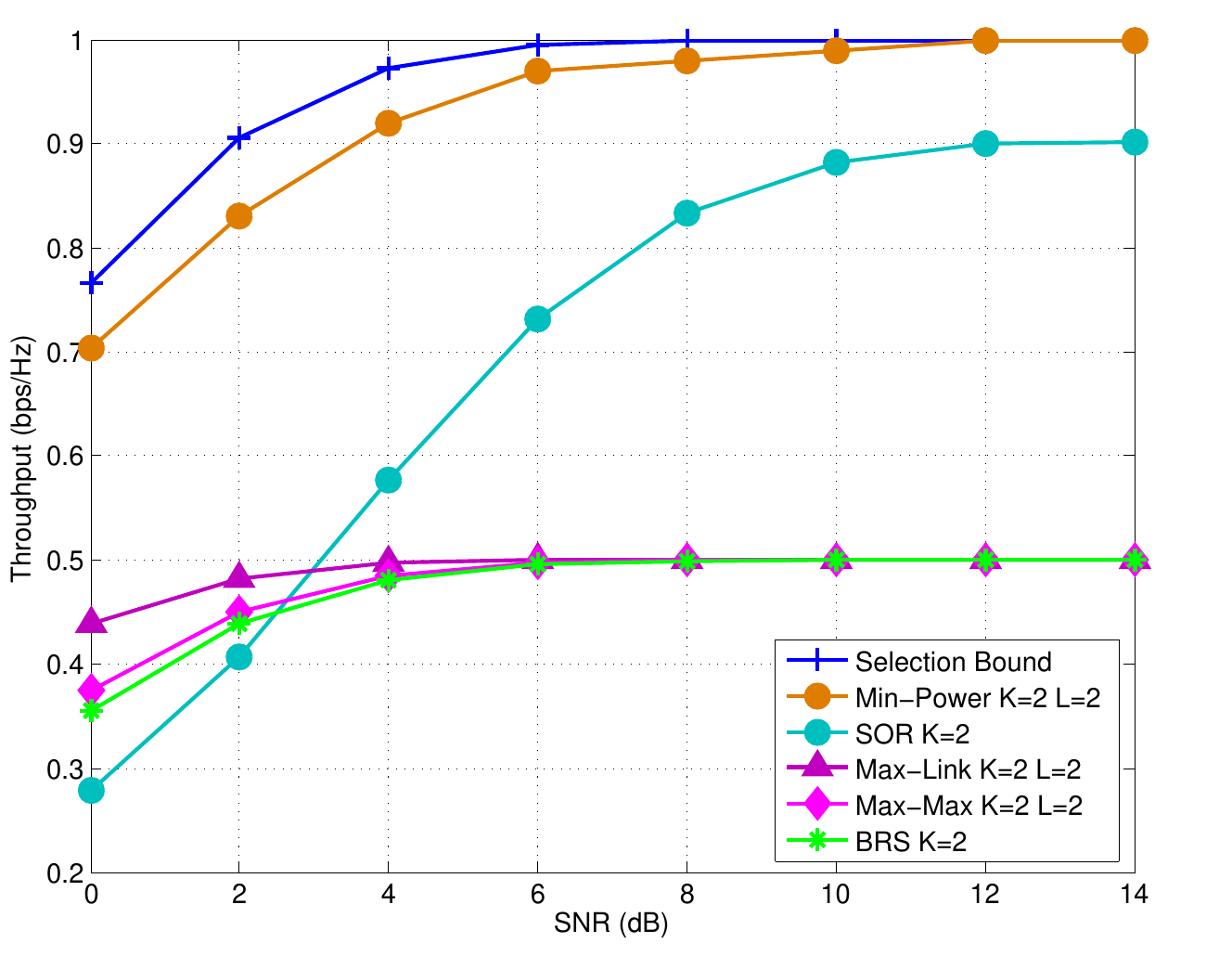}
\vspace{-0.48cm}
\caption{Average throughput for increasing transmit SNR for $K=L=2$. The upper-limit of $\minpow$ is more slowly reached than the half-duplex schemes since in the low SNR regime, single-link transmissions are more often performed.}
\vspace{-.6cm}
\label{thr_K2L2}
\end{figure}


Fig.~\ref{thr_K2Lv} illustrates average throughput for $\minpow$ as $L$ and transmission power increase. From the results we see that $\minpow$ for buffer-size above eight, follows the Selection Bound and their performance gap becomes negligible at about 8 dB. It is important to note that when the SNR is low, interference cancellation does not take place often and the proposed scheme chooses half-duplex transmissions instead of successive ones. This explains the gap between the Selection Bound and the cases of $L=100, \infty$ in the low SNR regime. Note that in this comparison we have $K=2$ relays and there is no flexibility in pair selection when a successive transmission is performed.
\begin{figure}[h]
\vspace{-0.4cm}
\centering
\includegraphics[width=8cm,height=6.7cm]{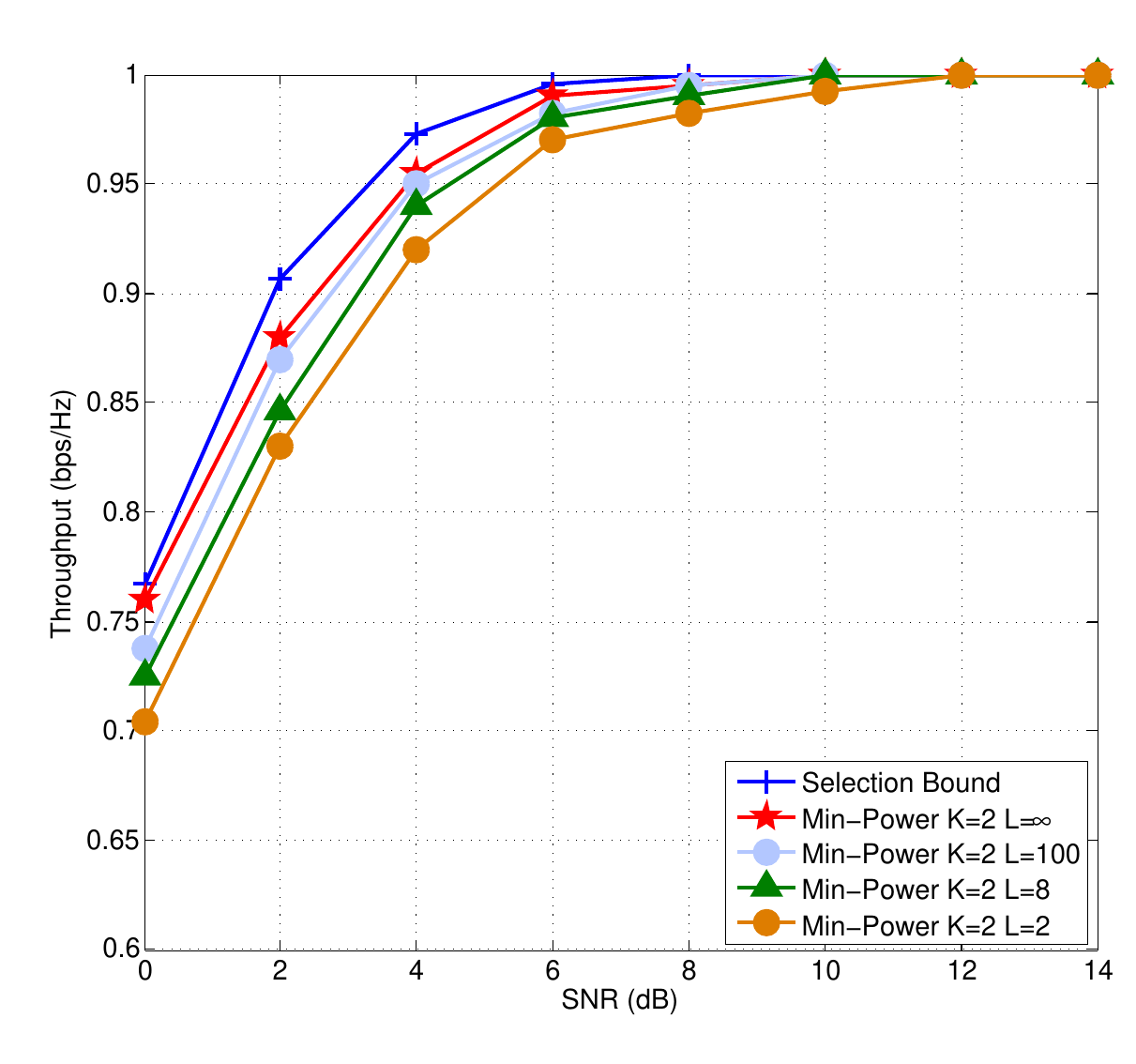}
\vspace{-0.48cm}
\caption{Average throughput for increasing transmit SNR for $K=2$ and varying $L$. For buffer-size above eight, $\minpow$ follows the Selection Bound and their gap becomes negligible at about 8 dB.}
\vspace{-.6cm}
\label{thr_K2Lv}
\end{figure}

The third parameter that we examine is the number of relays in the cluster. Fig.~\ref{thr_KvL4} shows the gain in average throughput as both $K$ and the transmission power increase. From the analysis of the $\minpow$ relay selection scheme, for each transmission we check $K(K-1)$ pairs to see whether or not interference cancellation can be performed. As we add more relays, the number of possible relay pairs increases from 2 in the case of $K=2$ to 6 in the case of $K=3$, while for $K=4$ we have 12 possible pairs. As a result, even for low SNR, successive transmissions are more possible for increasing $K$ and throughput tends to reach its maximum value more rapidly.
\begin{figure}[h]
\vspace{-0.2cm}
\centering
\includegraphics[width=8cm,height=6.7cm]{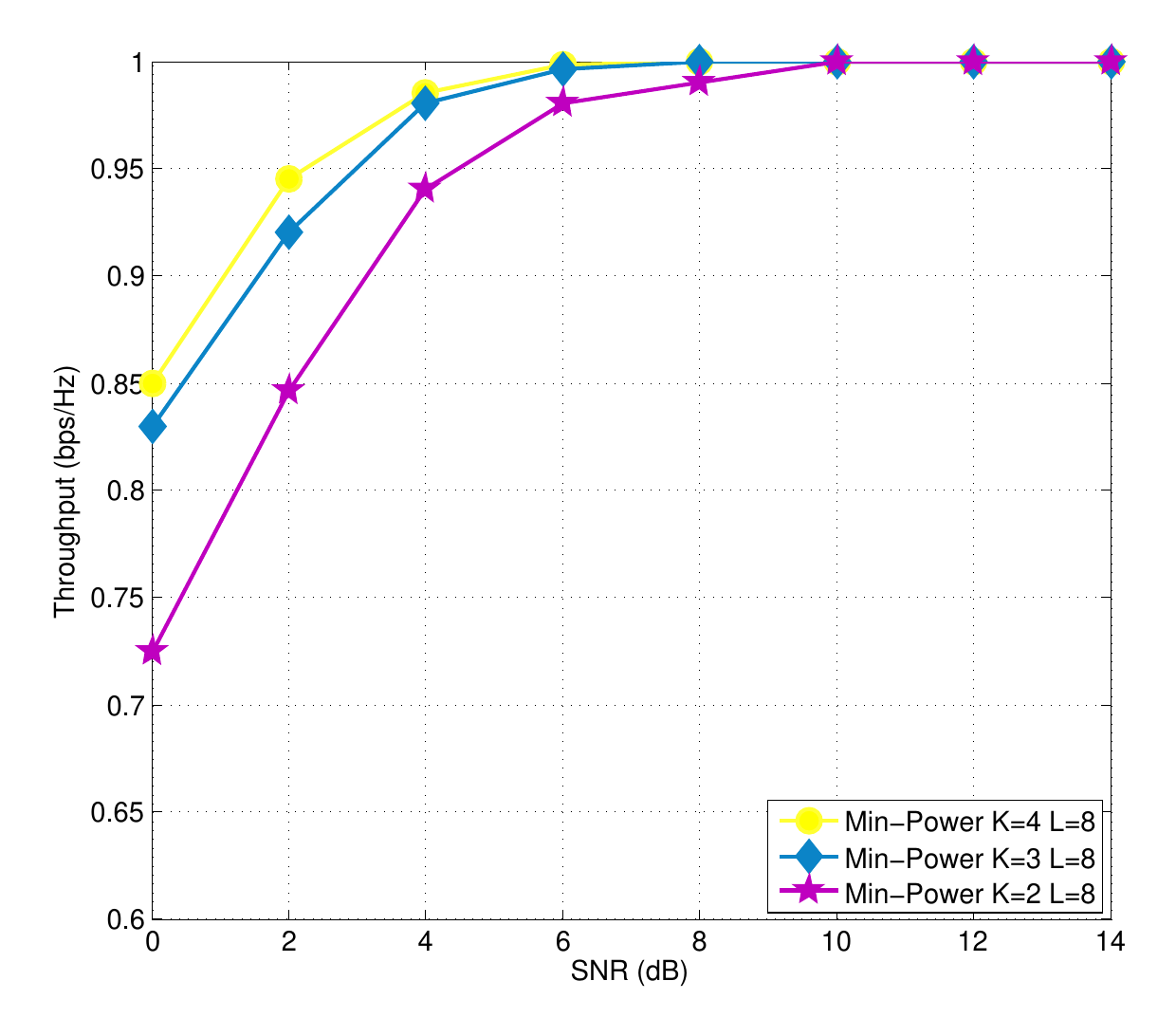}
\vspace{-0.48cm}
\caption{Average throughput for increasing transmit SNR for varying $K$ and $L=8$. As more relays are added, the possible relay increase and their number is equal to $K(K-1)$.}
\vspace{-.6cm}
\label{thr_KvL4}
\end{figure}


\subsection{Power Reduction}

The third metric that we study is the power reduction achieved by $\minpow$. The scheme used as a reference to compare the power gain is a buffer-aided relay selection scheme that does not employ power minimization in the selection process and transmits each time with the maximum available power. For example, if the fixed power scheme uses 6 dB to perform a transmission and $\minpow$ reduces this requirement to 2 dB we keep the difference of 4 dB and we calculate the average after the end of all transmissions for this step of maximum power value. As stated in Section~\ref{sec:minpow}, when $\minpow$ relay selection is used, a search for the pair of relays that requires the minimum sum of transmission powers starts for a given SNR threshold. Similarly, for the case of single-link transmission, the link requiring the minimum power to  achieve a successful transmission is selected. We note that when a successive transmission is feasible both the source and a relay transmit, resulting in increased power reduction margin as the most appropriate relay pair is selected. On the contrary, when a single-link transmission occurs, we calculate power reduction by comparing the power used by the transmitting node in the fixed power scheme to the power used by the transmitting node in the selected single-link. 

The first parameter that influences power reduction performance is $L$. Fig.~\ref{pow_K2Lv}, contains the curves for $K=2$ and various buffer sizes. Here the limiting factor is inter-relay interference. We see that for increased values of $L$ greater than 8, differences in power are minor but still, the Selection Bound is not met. This comes as a consequence of the difficulty to cancel IRI since only two relays are employed in the transmission. Note that the value on y-axis refers to the \emph{differential power gain}, i.e., the power in dB that is saved when the proposed scheme is employed. 

\begin{figure}[h]
\vspace{-0.4cm}
\centering
\includegraphics[width=8cm,height=6.7cm]{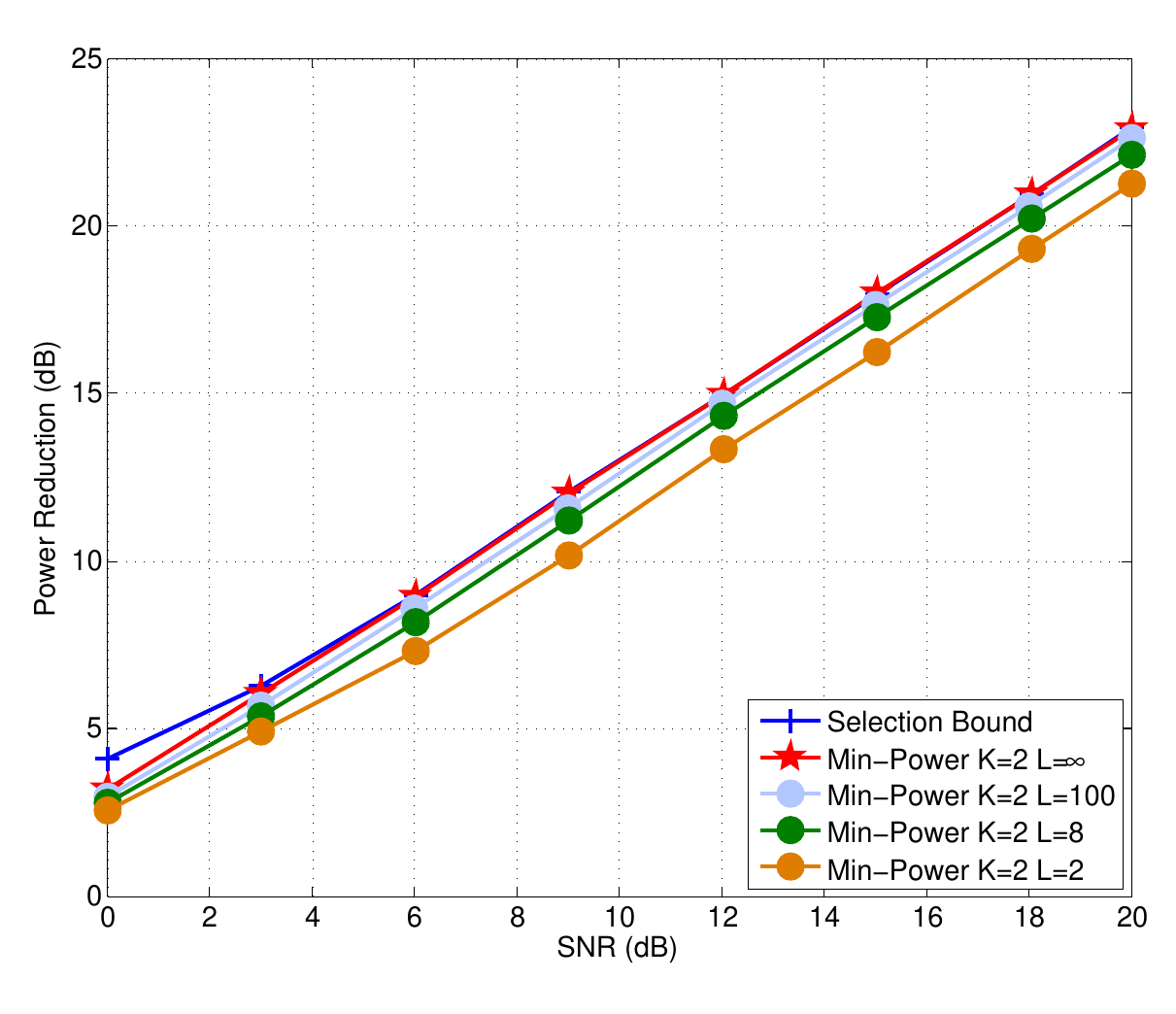}
\vspace{-0.48cm}
\caption{Power reduction for increasing transmit SNR for varying $L$ and $K=2$. For $L\geq8$, power differences are minor.}
\vspace{-.6cm}
\label{pow_K2Lv}
\end{figure}

In order to examine the effect of additional relays, we present the corresponding results in Fig.~\ref{pow_KvL8}. We can easily observe that adding more relays to the cluster, the relay selection alternatives increase, thus leading to improved power minimization. As fixed transmission rate is adopted, a saturation is observed since differences in the required power for successful transmissions are very close after $K=3$ relays.

\begin{figure}[h]
\vspace{-0.4cm}
\centering
\includegraphics[width=8cm,height=6.7cm]{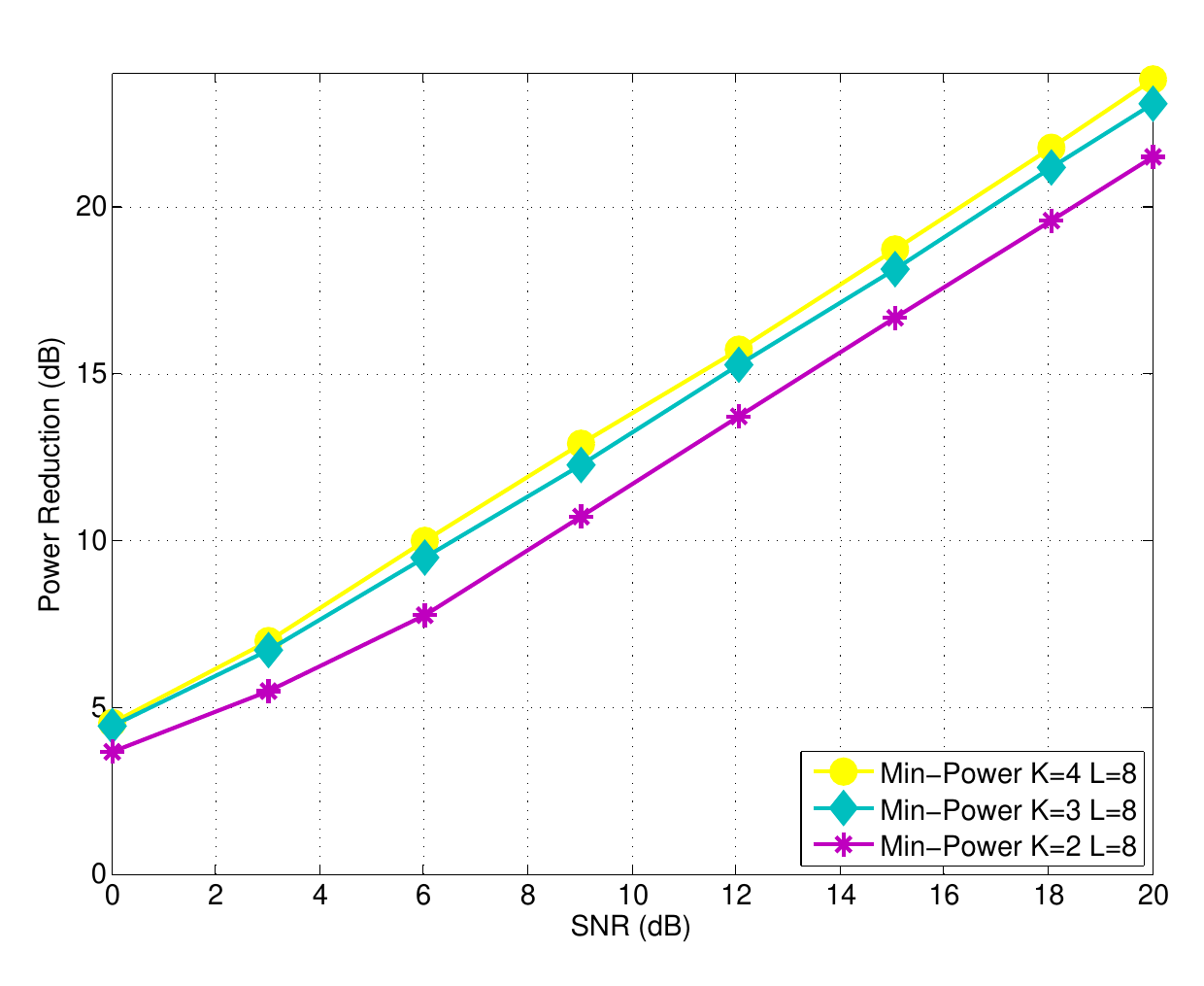}
\vspace{-0.48cm}
\caption{Power reduction for increasing transmit SNR for varying $K$ and $L=8$. Since fixed transmission rate is adopted, a saturation is observed as differences in the required power for transmission are very close after $K=3$ relays.}
\vspace{-.6cm}
\label{pow_KvL8}
\end{figure}

\subsection{Average Delay}

\begin{figure}[h]
\vspace{-0.2cm}
\centering
\includegraphics[width=8cm,height=6.7cm]{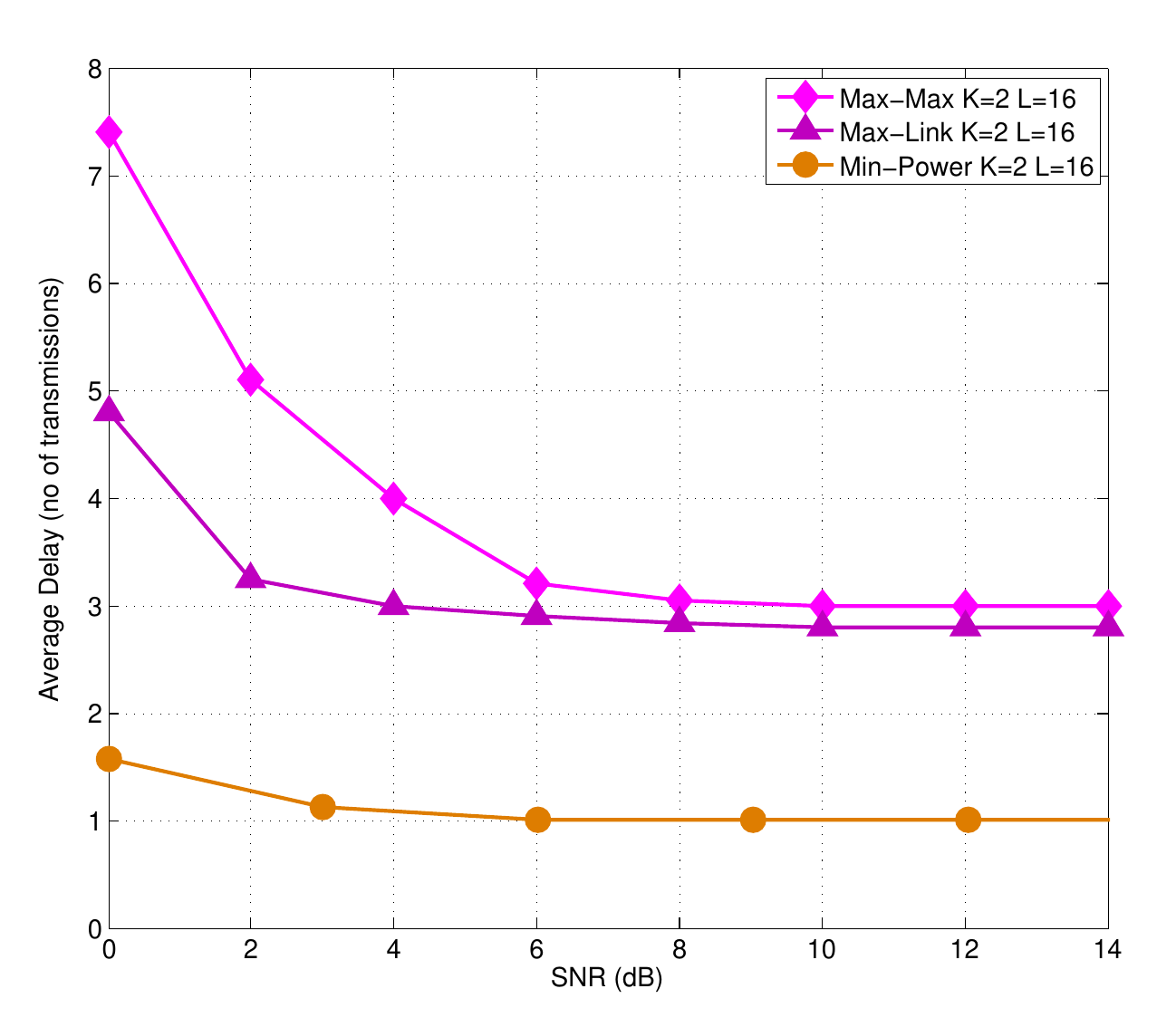}
\vspace{-0.48cm}
\caption{Average delay for increasing transmit SNR for $K=2$ and $L=16$. Due to the ability of two overlapping transmissions, the successive schemes achieve a full-duplex operation at high SNR.}
\vspace{-.6cm}
\label{delay}
\end{figure}

In the final set of comparisons we evaluate the average delay that each relay selection scheme incurs to the transmitted signals. From Fig.~\ref{delay}, we observe that in the low SNR regime there is a significant delay as single-link is the dominant transmission mode in that case. However, after 6dBs successive transmissions constitute the majority of the transmissions in the network and this is reflected on the average delay which is limited to only one transmission phase, thus achieving an almost full-duplex operation as the destination receives a new packet in each transmission phase. We note however, that for increased relay number and buffer size the average delay increases since some packets tend to remain in the relays' buffers for more transmission phases. In addition, the consideration of adaptive rate transmission may lead to better queue management by allowing relays to transmit with increased rate when an opportunity arises.
\vspace{-.2cm}

%
%
\section{Conclusions and Future Directions}\label{sec:conclusions}

Summarizing this paper, we proposed an opportunistic relaying protocol that minimizes the total power expenditure per time slot under an IRI cancellation scheme. Through power adaptation and buffer-aided relays, IRI is mitigated. It is the first time that interference cancellation is combined with buffer aided relays and power adaptation; the examples demonstrate the improvements achieved.
For the evaluation of the performance of the $\minpow$ relay selection policy, we performed comparisons with other schemes and the results showed that by combining successive transmission and a buffer-aided half-duplex protocol, gains were achieved in both outage and throughput performance.

Moreover, we studied the effect of buffer size and relay number on the outage, throughput and power reduction metrics. As buffer size and relays increase, the gain obtained by $\minpow$ vanishes. This observation suggests the need for additional approaches that will improve the potential of $\minpow$. A future direction includes the study of network coding schemes and exploitation of the IRI by using superposition coding. In addition, a major concern raised by $\minpow$ relay-pair selection is the additional delay introduced into the network. Compared to \cite{IKH2} and \cite{KRICHAR}, $\minpow$ allows two simultaneous transmissions and thus, has the potential to reduce the delay observed by these schemes. In addition, compared to \cite{IKH3}, the consideration of IRI in the selection process allows optimal relay-pair selection and a future combination with adaptive rate transmission can lead our scheme to a more efficient queue management by allowing the relays to receive and transmit with increased rates when the opportunity arises. This could possibly contribute to a reduction in the number of transmission phases that packets reside in the buffers of the relays.  As we considered a topology with  i.i.d. links, extending $\minpow$ with rate adaptation capabilities can leverage its deficits in non i.i.d. scenarios where fixed rate transmissions limit its performance. Finally, delay-aware characteristics should be implemented in the algorithm, where the selection process will prioritize relays which may experience packet losses due to excess delay.

%
%

%
%
\end{document}